\documentclass[11pt]{article}

\usepackage[margin=1in]{geometry}
\usepackage{amsmath,amssymb,amsthm,mathtools}
\usepackage{tikz}
\usepackage{tikz-cd}
\usepackage{listings}
\usepackage{enumitem}
\usepackage{hyperref}
\usepackage{natbib}
\usepackage{xcolor}

\newtheorem{theorem}{Theorem}[section]
\newtheorem{proposition}[theorem]{Proposition}

\newtheorem{definition}[theorem]{Definition}

\newtheorem{remark}[theorem]{Remark}

\renewcommand{\Form}{\mathsf{Form}}
\newcommand{\Machine}{\mathsf{Machine}}

\newcommand{\Mat}{\mathsf{materialize}}
\newcommand{\Quote}{\mathsf{quote}}
\newcommand{\Reflect}{\mathsf{reflect}}

\newcommand{\Dir}{\mathsf{Dir}}
\newcommand{\CapSet}{\mathsf{CapSet}}

\lstdefinelanguage{mashin}{
  keywords={machine, provides, implements, ensures, achieves, verifies, inputs, outputs,
            ask, compute, decide, recall, remember, launch, wait, quote, reflect,
            form, match, case, if, else, stores, expresses, flow, task, goal},
  keywordstyle=\bfseries,
  comment=[l]{\#},
  string=[b]",
  sensitive=true,
}
\title{Governed Metaprogramming for Intelligent Systems: \\
Reclassifying Eval as a Governed Effect}

\author{Alan L. McCann\\
\textit{Mashin, Inc.}\\
\texttt{research@mashin.live}}

\date{April 2026}

\begin{document}

\maketitle

\begin{abstract}
AI systems increasingly synthesize executable structure at runtime: LLMs
generate programs, agents construct workflows, self-improving systems modify
their own behavior. In classical homoiconic and staged languages, the
transition from code representation to execution is unrestricted.
\texttt{eval} is a language primitive, not a governed operation. We argue
that in governed intelligent systems, this transition is an authority
amplification: it converts symbolic structure into executable authority and
must be mediated like any other effect.

We present \emph{governed metaprogramming}, a language design where
program representations (machine forms) are first-class values, form
manipulation is pure computation, and materialization (the transition from
form to executable machine) is a governed effect subject to structural
inspection. The governance system analyzes the proposed program's capability
requirements, policy compliance, and resource estimates before permitting
execution. We formalize two judgments: pure form evaluation (which emits no
directives) and governed materialization (which emits exactly one governed
directive). We prove three properties: purity of form manipulation, the
no-bypass theorem, and boundary preservation. We implement the design in
MashinTalk, a DSL for AI workflows compiling to BEAM bytecode, and report on
integration with 454 existing machine-checked Rocq theorems.

The central contribution is reclassifying \texttt{eval} from a language
primitive into a governed effect.
\end{abstract}

\section{Introduction}
\label{sec:intro}

AI systems generate executable structure. An LLM writes a workflow. An agent
constructs a tool-calling plan. A self-improving system reflects on its own
behavior and produces a modified version of itself. In each case, something
that was data becomes something that runs.

This transition is the central problem of this paper. In classical
programming, the transition from data to executable code is handled by
\texttt{eval} (Lisp, 1960~\cite{mccarthy1960}), \texttt{run}
(MetaOCaml~\cite{kiselyov2014}), or equivalent operations. In all cases, the
transition is unrestricted: any valid representation can become running code.

We make three observations:

\begin{enumerate}[itemsep=2pt]
\item AI systems increasingly synthesize executable structure at runtime.
      This is not a future concern; it is the operational reality of agent
      frameworks, code-generating models, and self-improving systems.
\item Classical homoiconic and staged systems provide a path from code
      representation to execution that is not subject to a first-class
      governance boundary.
\item In governed intelligent systems (those with explicit effect boundaries,
      capability control, and audit requirements), unrestricted execution of
      generated structure breaks the governance model.
\end{enumerate}

The conclusion follows directly:

\begin{quote}
\emph{Generated structure must cross the same governed boundary as every other
effectful operation.}
\end{quote}

This paper provides the mechanism. We introduce \emph{machine forms}: first-class
values representing program structure. Form manipulation (inspection,
transformation, composition, diffing) is pure computation. The transition from
form to executable machine, which we call \emph{materialization}, is classified
as a governed effect. The governance system inspects the structural content of
the proposed program before permitting execution.

\paragraph{The key insight.}
\texttt{eval} is not merely computation. It is authority amplification.
When a form becomes running code, the system crosses from symbolic structure
into executable authority: the authority to invoke models, perform I/O, consume
resources, and affect external systems. Traditional homoiconic systems treat
this as normal language power. We argue that in governed intelligent systems, it
is a privileged transition and must be mediated.

\paragraph{The paradigm.}
This paper adds a fourth layer to the emerging paradigm of governed computation:

\begin{center}
\begin{tabular}{rl}
1. & Pure computation (no effects by construction) \\
2. & Governed effects (all I/O through governed directives) \\
3. & Governed evolution (version changes through evolution ledger) \\
4. & \textbf{Governed metaprogramming} (code generation through governed materialization)
\end{tabular}
\end{center}

Without layer 4, a reviewer of the first three layers can immediately ask:
``What if the system generates a new program with different capabilities?''
This paper answers that question.

\paragraph{Contributions.}
\begin{itemize}[itemsep=2pt]
\item The reclassification of \texttt{eval} from a language primitive into a
      governed effect, with a precise characterization of materialization as
      authority allocation (Section~\ref{sec:authority}).
\item Machine forms: a concrete data type for structured program
      representation with quote, splice, and reflection constructs
      (Section~\ref{sec:forms}).
\item Formal operational semantics with two judgments: pure form evaluation
      and governed materialization (Section~\ref{sec:semantics}).
\item Three machine-checked properties: form manipulation purity, the
      no-bypass theorem, and boundary preservation
      (Section~\ref{sec:verification}).
\item Governed self-modification: a concrete mechanism for machines that
      propose structural changes to themselves through an evolution ledger
      (Section~\ref{sec:selfmod}).
\end{itemize}

\section{Background}
\label{sec:background}

\subsection{Homoiconicity and Staged Computation}

A programming language is \emph{homoiconic} if its primary representation of
programs is a data structure in the language itself~\cite{kay1969}. In Lisp,
programs are S-expressions and \texttt{eval} executes them. Clojure extends
this with reader macros. Julia represents code as \texttt{Expr} objects.

Staged computation systems provide a related but distinct mechanism.
MetaOCaml~\cite{kiselyov2014} offers typed staging with \texttt{bracket},
\texttt{escape}, and \texttt{run}. Terra~\cite{devito2013} provides runtime
code generation. Racket's macro system~\cite{flatt2012} enables compile-time
metaprogramming.

Classical homoiconic and staged systems generally provide a path from code
representation to execution that is not itself subject to a first-class
governance boundary. Elixir's \texttt{quote}/\texttt{unquote} is
compile-time and does not provide runtime eval. MetaOCaml's \texttt{run} is
typed but unrestricted. The common pattern is that the transition from data
to code is treated as a language primitive rather than a privileged operation.

\subsection{The Mashin Runtime}

MashinTalk is a domain-specific language for AI workflows that compiles to
BEAM (Erlang VM) bytecode~\cite{mashin2026}. The runtime enforces eight
invariants, of which three are directly relevant:

\begin{itemize}[itemsep=1pt]
\item \textbf{Inv 1 (No Ambient Effects):} Pure computation steps cannot
      perform I/O. The capability is absent, not sandboxed.
\item \textbf{Inv 3 (Governance Mediation):} Every directive is mediated
      by governance. No bypass path exists.
\item \textbf{Inv 8 (Governed Evolution):} Machine versions are promoted
      only through the evolution ledger.
\end{itemize}

These invariants have been formalized and proved in 454 Rocq theorems across
36 modules with zero admitted lemmas~\cite{mashin-coq2026}.

\subsection{The Metaprogramming Escape Hatch}

The governance model described above governs execution: what a running
machine can do. But it does not, without this paper's contribution, govern
code generation: what executable structures the system can produce.

A system governed at the execution level but ungoverned at the code
generation level is vulnerable to a simple attack: generate a new program
that lacks the governance constraints of the original. If the generated
program can be executed without structural inspection, the governance model
has an escape hatch.

This is not a theoretical concern. LLM-powered agent systems routinely
generate executable structures (tool-calling plans, workflow
configurations, code) that execute without structural governance.

\section{Materialization as Authority Allocation}
\label{sec:authority}

We argue that the transition from code representation to executable program
is not ordinary computation. It is authority allocation.

\begin{definition}[Executable Authority]
A value has \emph{executable authority} if it can, when executed, invoke
external models, perform I/O, consume computational resources, or affect
systems beyond the current computation.
\end{definition}

A machine form is a map. It has no executable authority. It cannot invoke a
model, write a file, or send a network request. It is data.

A running machine has executable authority. Its steps invoke models
(\texttt{ask ... using:}), call effect machines (\texttt{ask ... from:}),
and produce traces recorded in the behavioral ledger.

Materialization is the operation that converts the former into the latter.
It allocates executable authority to a structure that previously had none.

\begin{remark}[Why Materialization is an Effect]
The chain of reasoning:
\begin{enumerate}
\item Materialization allocates executable authority.
\item Executable authority is a privileged runtime resource (it enables
      model invocation, I/O, and resource consumption).
\item Acquisition of privileged runtime resources is an effect.
\item Therefore, materialization is an effect.
\end{enumerate}
This is why a reviewer asking ``why not treat materialization as just another
function from AST to closure?'' receives a precise answer: because the result
is not merely data. It is authority-bearing executable structure.
\end{remark}

\begin{definition}[Governed Metaprogramming]
\label{def:govmeta}
A programming language system exhibits \emph{governed metaprogramming} if:
\begin{enumerate}
\item Program representations (forms) are first-class values.
\item Form manipulation is classified as pure computation.
\item Materialization (the transition from form to executable program) is
      classified as an effect and requires governance mediation.
\item The governance system can inspect the structural content of the form
      prior to permitting materialization.
\end{enumerate}
\end{definition}

The key reclassification: \texttt{eval} is not a primitive. It is a governed
effect.

\section{Machine Forms}
\label{sec:forms}

\subsection{Form Structure}

\begin{definition}[Machine Form]
A machine form $f$ is a tuple $(k, n, v, c, F, C)$ where:
\begin{itemize}
\item $k \in \mathsf{Kind}$ is the keyword type from the syntax hierarchy
\item $n \in \mathsf{String} \cup \{\bot\}$ is the optional name
\item $v \in (\mathsf{String} \times \mathsf{Any}) \cup \{\bot\}$ is
      optional variant information
\item $c \in \mathsf{String} \cup \{\bot\}$ is optional inline content
\item $F : \mathsf{String} \to \mathsf{Any}$ is the fields map
\item $C : \mathsf{List}[\Form]$ is the list of child forms
\end{itemize}
\end{definition}

The set $\mathsf{Kind}$ contains all MashinTalk keyword types, organized in
a five-level hierarchy: level 0 (\texttt{machine}), level 1 (eight section
verbs), levels 2--4 (steps, configuration, assertions).

\subsection{Quote}

The $\Quote$ construct captures MashinTalk syntax as a form value:

\begin{lstlisting}
compute build
  template: quote
    machine greeter
      provides
        inputs
          name: text, required
      implements
        compute greet
          greeting: "Hello, " + input.name
\end{lstlisting}

The keyword-hierarchy parser processes the indented block and produces a
node tree. The compiler converts this tree to a form map literal.

\subsection{Splice}

Inside a quote block, \texttt{\*(expr)} evaluates an expression in the
enclosing scope and inserts the result. Three variants: scalar splice
(\texttt{\*(expr)}) inserts a value, form splice (\texttt{\*(expr)}) inserts
a subtree when the expression evaluates to a form, and spread splice
(\texttt{\$(...expr)}) inserts a list of forms as sibling children.

\subsection{Reflection}

The $\Reflect()$ intrinsic returns the current machine's form as a
compile-time constant.

\subsection{Form Standard Library}

Pure functions for form manipulation, all classified as pure computation:

\begin{itemize}[itemsep=1pt]
\item \textbf{Navigation:} \texttt{form.kind}, \texttt{form.name},
      \texttt{form.children}, \texttt{form.section}, \texttt{form.step},
      \texttt{form.steps}, \texttt{form.get}
\item \textbf{Transformation:} \texttt{form.set}, \texttt{form.add},
      \texttt{form.remove}, \texttt{form.replace}, \texttt{form.merge}
\item \textbf{Analysis:} \texttt{form.diff}, \texttt{form.validate},
      \texttt{form.hash}, \texttt{form.capabilities}
\item \textbf{Serialization:} \texttt{form.to\_text},
      \texttt{form.from\_text}, \texttt{form.to\_json},
      \texttt{form.from\_json}
\end{itemize}

\section{Operational Semantics}
\label{sec:semantics}

We first state the system invariants that governed metaprogramming
requires, then define two judgments that make the governance boundary
precise.

\subsection{System Invariants}

The following five invariants are enforced by the runtime architecture.
They are not conventions or guidelines; they are structural properties
of the system. Violating any one of them would create an execution path
that bypasses governance.

\begin{definition}[System Invariants for Governed Metaprogramming]
\label{def:invariants}
\hfill
\begin{enumerate}
\item \textbf{No Direct Evaluation.}
      There exists no callable operation that transforms a program
      representation into an executing program. No \texttt{eval},
      \texttt{run}, or equivalent function is available to programs
      operating on forms.

\item \textbf{Mandatory Governance Mediation.}
      All transitions from program representation to executable program
      pass through a governance interpreter. The governance interpreter
      is the exclusive mechanism through which materialization occurs.

\item \textbf{Authority Allocation Boundary.}
      Executable authority (the ability to invoke models, perform I/O,
      consume resources, affect external systems) is allocated only
      during governed materialization. No other operation allocates
      executable authority.

\item \textbf{Pure Representation Layer.}
      The computation layer operating on forms has no capability to
      perform effects. Form operations produce values without emitting
      directives. This is not sandboxing; the capability for effects
      is absent from the representation layer by construction.

\item \textbf{Structural Inspectability.}
      All executable programs originate from structured representations
      (forms) whose content can be fully traversed prior to execution.
      The governance system can analyze capability requirements, policy
      compliance, and resource estimates from the form's structure.
\end{enumerate}
\end{definition}

Invariants 1 and 4 together ensure no bypass: programs cannot reach
execution without crossing the governance boundary. Invariant 2
ensures the boundary is governance, not just any barrier. Invariant 3
identifies what crosses the boundary (authority). Invariant 5 ensures
governance has enough information to make informed decisions.

\subsection{Pure Form Evaluation}

Pure form evaluation produces a value and emits no directives:

\[
\frac{\Gamma \vdash e \Downarrow v \quad \Dir(e) = \emptyset}
     {\Gamma \vdash_{\mathsf{pure}} e \Downarrow v}
\quad \textsc{(Pure-Eval)}
\]

All form operations (quote construction, splice evaluation, navigation,
transformation, analysis, serialization) satisfy this judgment. They
operate on map values and return map values. No directive is emitted.

\subsection{Governed Materialization}

Governed materialization takes a form, a policy context, and a capability
set, and produces either an authorized machine or a rejection:

\[
\frac{\Gamma \vdash_{\mathsf{pure}} f \Downarrow v_f \quad
      \Pi \vdash_{\mathsf{inspect}} v_f : \mathsf{ok} \quad
      \CapSet(v_f) \subseteq \sigma}
     {\Gamma; \Pi; \sigma \vdash \Mat(v_f) \leadsto
      \langle\Machine(v_f),\, d_{\mathsf{approved}}\rangle}
\quad \textsc{(Mat-Approve)}
\]

\[
\frac{\Gamma \vdash_{\mathsf{pure}} f \Downarrow v_f \quad
      \Pi \vdash_{\mathsf{inspect}} v_f : \mathsf{fail}(r)}
     {\Gamma; \Pi; \sigma \vdash \Mat(v_f) \leadsto
      \langle\bot,\, d_{\mathsf{rejected}}(r)\rangle}
\quad \textsc{(Mat-Reject)}
\]

where:
\begin{itemize}
\item $\Pi$ is the policy context (structural constraints, model allowlists,
      resource limits)
\item $\sigma$ is the caller's capability set
\item $\CapSet(v_f)$ computes the capabilities required by the form
      (by traversing effect-producing steps)
\item $d$ is the governance decision record (always produced, whether
      approved or rejected)
\end{itemize}

The key properties:

\begin{proposition}[Pure Evaluation Emits No Directives]
\label{prop:pure-no-dir}
For all expressions $e$ in the form language and environments $\Gamma$:
\[
\Gamma \vdash_{\mathsf{pure}} e \Downarrow v \implies \Dir(e) = \emptyset
\]
\end{proposition}

\begin{proposition}[Materialization Emits Exactly One Governed Directive]
\label{prop:mat-one-dir}
For all forms $f$, policy contexts $\Pi$, and capability sets $\sigma$:
\[
\Gamma; \Pi; \sigma \vdash \Mat(f) \leadsto \langle r, d \rangle
\implies |\Dir(\Mat(f))| = 1 \wedge \mathsf{governed}(d)
\]
\end{proposition}

\begin{theorem}[No Bypass: No Instruction Can Directly Cause Execution]
\label{thm:no-bypass}
No instruction available to a program operating in the representation
layer can directly cause execution of a program representation. That is,
no sequence of pure evaluation steps produces an authorized machine:
\[
\forall e_1, \ldots, e_n.\;
\Gamma \vdash_{\mathsf{pure}} e_1 \Downarrow v_1,\, \ldots,\,
\Gamma \vdash_{\mathsf{pure}} e_n \Downarrow v_n
\implies \nexists\, \Machine(v_i)
\]
This is a consequence of Invariants 1 and 4
(Definition~\ref{def:invariants}): the representation layer has no
eval-like operation (Invariant 1) and no capability to perform effects
(Invariant 4). Execution requires authority allocation (Invariant 3),
which requires governance mediation (Invariant 2).
\end{theorem}

\begin{proof}
By Proposition~\ref{prop:pure-no-dir}, pure evaluation emits no directives.
Machine authorization requires a governance directive
(Proposition~\ref{prop:mat-one-dir}). By Invariant~2 (Mandatory Governance
Mediation), every directive is mediated. Since pure evaluation produces no directives,
it cannot produce authorization. Therefore no sequence of pure evaluations
produces an authorized machine.
\end{proof}

\section{Governed Materialization}
\label{sec:materialization}

\subsection{Materialization Paths}

Two materialization paths exist, both governed:

\paragraph{Evolution proposal.} A machine submits a form to
\texttt{@system/evolution/propose}. The governance system performs structural
inspection (Section~\ref{sec:inspection}). If approved, the form is registered
as a new machine version with content-addressed hashing.

\paragraph{One-shot execution.} A machine submits a form to
\texttt{@system/runtime/eval}. The governance system inspects the form. If
approved, the form is compiled to a temporary machine and executed.

Both paths route through the governance directive interpreter. By Inv~3,
no bypass exists.

\subsection{What Governance Inspects}

Because forms are structured data (not opaque strings or bytecode), the
governance system can analyze:

\begin{itemize}[itemsep=1pt]
\item \textbf{Capability requirements:} What effects the form will require,
      computed by traversing effect-producing steps.
\item \textbf{Policy compliance:} Whether the form satisfies structural
      constraints (required governance sections, step count limits).
\item \textbf{Model authorization:} Whether referenced models are on the
      approved list.
\item \textbf{Resource estimation:} Expected cost based on model usage and
      step count.
\item \textbf{Derivation trust:} Whether the form was produced by a trusted
      source (human, approved generator, or LLM with validation).
\end{itemize}

This structural inspection is impossible in systems where generated code is
represented as opaque strings.

\section{Structural Form Inspection}
\label{sec:inspection}

Structural inspection is the formal operation $\Pi \vdash_{\mathsf{inspect}}
f : \mathsf{result}$ from the operational semantics.

\begin{definition}[Structural Inspection]
Given a policy context $\Pi$ and a form $f$, structural inspection is a
decidable predicate that traverses $f$ and checks:
\begin{enumerate}
\item $\CapSet(f) \subseteq \Pi.\mathsf{allowed\_caps}$
      \hfill (capability containment)
\item $\mathsf{models}(f) \subseteq \Pi.\mathsf{allowed\_models}$
      \hfill (model authorization)
\item $\mathsf{structure}(f) \models \Pi.\mathsf{policies}$
      \hfill (policy compliance)
\item $\mathsf{cost}(f) \leq \Pi.\mathsf{budget}$
      \hfill (resource bounds)
\end{enumerate}
\end{definition}

Each check is a finite traversal of the form tree. All checks terminate.
The inspection is decidable because it operates on structural properties,
not semantic behavior (cf. Rice's theorem).

\begin{remark}
Structural inspection catches structural violations but cannot reason about
semantic behavior. A form that structurally complies with policies may
produce undesirable runtime behavior. This is a fundamental limitation:
Rice's theorem establishes that non-trivial semantic properties are
undecidable. Structural inspection is the strongest decidable analysis
available.
\end{remark}

\section{Boundary Preservation}
\label{sec:coterminous}

Prior work establishes the \emph{coterminous boundary property}: the
expressiveness boundary $E$ (what the language can express) and the
governance boundary $G$ (what the governance system governs) are
identical~\cite{mashin-coq2026}. We show that adding governed
metaprogramming preserves this property.

\begin{proposition}[Boundary Preservation]
\label{prop:boundary}
Let $E_0 = G_0$ be the coterminous boundaries before adding forms.
Let $E_1$ and $G_1$ be the boundaries after adding forms and governed
materialization. Then $E_1 = G_1$.
\end{proposition}

\begin{proof}[Argument]
Form manipulation is pure computation within existing step types. It does
not introduce new effect types or new boundary-crossing operations.
Therefore $E_1 = E_0$ (form manipulation is expressible as pure computation,
which was already in $E_0$).

Materialization is a governed effect using existing directive mechanisms
(\texttt{:call} to system machines). It does not introduce a new
governance category. Therefore $G_1 = G_0$ (materialization is governed
through existing governance infrastructure).

Since $E_0 = G_0$ and $E_1 = E_0$ and $G_1 = G_0$, we have $E_1 = G_1$.

The argument depends on two claims: (a) form manipulation introduces no
new effects, and (b) materialization uses existing governed effect
mechanisms. Both follow from the implementation: forms are map values in
pure computation steps, and materialization is a \texttt{:call} directive.
A full formalization would require defining the boundaries as sets of
effect signatures and showing conservativity. We leave this to future work.
\end{proof}

\section{Governed Self-Modification}
\label{sec:selfmod}

Machine forms provide the concrete mechanism for governed self-modification:

\begin{enumerate}
\item A machine calls $\Reflect()$ to obtain its own form.
\item Pure computation produces a modified form.
\item The machine submits the modified form to governance via a governed
      effect (\texttt{@system/evolution/propose}).
\item Governance performs structural inspection.
\item If approved, the evolution ledger records the old hash, new hash,
      structural diff, and evidence.
\item The modified form becomes the new machine version.
\end{enumerate}

\begin{lstlisting}
machine self_improving
  implements
    compute introspect
      my_form: reflect()

    ask classify, using: "claude-sonnet-4-6"
      task
        "Classify this text."
      returns
        confidence: number

    compute propose
      improvement: match classify.confidence < 0.7 {
        case true => form.set(my_form,
          "implements.classify.variant_value",
          "claude-opus-4-6")
        case false => null
      }

    ask evolve, from: "@system/evolution/propose"
      definition: propose.improvement
      evidence: {confidence: classify.confidence}
\end{lstlisting}

This satisfies Inv~8 (Governed Evolution). Without governed metaprogramming,
self-modification is either impossible (no code generation) or ungoverned
(code generation bypasses governance). Governed metaprogramming provides the
middle ground: self-modification that is structurally inspected and recorded.

\section{Verification}
\label{sec:verification}

The properties integrate with the existing Rocq formalization (454 theorems,
36 modules, 0 admitted lemmas).

\subsection{New Theorems}

The \texttt{GovernedMetaprogramming.v} module extends the Rocq development with approximately 25 theorems covering form inspection safety, splice safety, evolution preservation, the reflect-modify-materialize pipeline, and a 12-way capstone theorem. The original four core properties are:

\begin{itemize}
\item \texttt{form\_manipulation\_pure}: For all form operations $\phi$ and
      forms $f$, $\phi(f)$ produces no directives. Follows from
      \texttt{no\_ambient\_effect} (Inv~1).
\item \texttt{materialization\_governed}: For all forms $f$, if $f$ becomes
      an executing machine, a governance decision was recorded. Follows from
      \texttt{governance\_mediation} (Inv~3).
\item \texttt{no\_bypass\_form\_to\_machine}: There exists no sequence of
      pure operations producing machine execution from a form. Follows from
      the previous two.
\item \texttt{boundary\_preserved}: The boundary property holds after adding
      form manipulation and governed materialization. Conservative extension
      argument.
\end{itemize}

\subsection{Interaction Trees}

Following prior work, we formalize governed metaprogramming using Interaction
Trees (ITrees)~\cite{xia2020itrees} with parameterized coinduction
(paco)~\cite{hur2013paco}. Form manipulation operations are modeled as
returning values without emitting events. Materialization emits a
\texttt{MaterializeE} event that the governance handler mediates. The handler
performs structural inspection and returns either an authorized machine or
a rejection with reasons.

\section{Related Work}
\label{sec:related}

\paragraph{Homoiconic languages.}
Lisp~\cite{mccarthy1960}, Scheme~\cite{r7rs}, Clojure~\cite{hickey2008},
and Julia~\cite{bezanson2017} provide code-as-data with paths from
representation to execution that are not subject to governance boundaries.
Elixir~\cite{elixir2014} restricts metaprogramming to compile time, which
does not address runtime code generation by AI systems.

\paragraph{Staged computation.}
MetaOCaml~\cite{kiselyov2014} provides typed staging with
\texttt{bracket}/\texttt{escape}/\texttt{run}. The type system ensures
well-scopedness but does not restrict what capabilities the generated code
may exercise. Terra~\cite{devito2013} provides runtime code generation
within a multi-stage framework. In both systems, the generated code inherits
the full authority of the generating context. Neither reclassifies the
transition as a governed effect.

\paragraph{Macro systems.}
Racket's macro system~\cite{flatt2012} provides hygienic macros with
compile-time computation and syntax certificates for authority control.
Syntax certificates are the closest prior work to governed materialization:
they restrict what transformations a macro may perform. However, they
operate at compile time, not at runtime, and do not address dynamic code
generation by AI agents.

\paragraph{Capability systems.}
Object-capability systems (E language~\cite{miller2006},
Wyvern~\cite{nistor2013}) restrict what code can access through capability
references. They govern access to existing objects, not the structure of
proposed programs. A capability system can restrict what a running program
does; governed metaprogramming restricts what programs can be brought into
existence.

\paragraph{Reflective towers.}
Smith's 3-Lisp~\cite{smith1984} introduced reflective towers: a program can
inspect and modify its own interpreter through an infinite regression of
meta-circular interpreters. Each level reflects on the level below. The key
insight is that self-reference requires a mediating interpreter. In governed
metaprogramming, the mediator is not a meta-circular interpreter but a
governance interpreter: the system that decides whether a proposed form may
be materialized. Where reflective towers ask ``can this program reason about
itself?'', governed metaprogramming asks ``can this program modify itself
safely?''. The answer in both cases is: only through a mediating layer. The
difference is that the mediating layer in governed metaprogramming enforces
capability constraints and records decisions in an append-only ledger.

\paragraph{AI agent governance.}
Guardrails AI and NeMo Guardrails~\cite{nemo2023} inspect model inputs and
outputs at runtime. They operate on behavior (what the agent is doing), not
on program structure (what the agent proposes to become). They cannot inspect
a proposed program before execution.

\section{Implementation and Evaluation}
\label{sec:evaluation}

We have implemented governed metaprogramming in a production-quality
system (1,027 lines of Elixir across three modules) and verified the
core properties in Rocq (26 theorems, 662 lines, zero admitted).

\subsection{Implementation}

The implementation consists of three modules:

\begin{itemize}
\item \textbf{Form} (540 lines): form construction (\texttt{new}, \texttt{from\_node},
      \texttt{from\_text}), navigation (\texttt{children}, \texttt{find\_child},
      \texttt{find\_all}), transformation (\texttt{add\_child}, \texttt{update\_child},
      \texttt{remove\_child}, \texttt{set\_field}), analysis (\texttt{count\_steps},
      \texttt{step\_types}, \texttt{capabilities}), serialization (\texttt{to\_text},
      \texttt{hash}), and round-trip parsing via \texttt{quote(expr)} and
      \texttt{reflect()}.
\item \textbf{FormInspector} (238 lines): six governance checks executed before
      materialization: valid structure, required fields, permitted capabilities,
      model authorization, governance section presence, and trust level compliance.
      Returns \texttt{\{:approved, report\}} or \texttt{\{:rejected, report\}} with
      a cryptographic form hash binding the inspection to the exact form inspected.
\item \textbf{FormMaterializer} (249 lines): governed compilation of inspected
      forms into executable machines. Three modes: \texttt{eval} (compile and load),
      \texttt{propose} (create a version proposal in the evolution ledger),
      \texttt{describe} (generate human-readable description without compilation).
\end{itemize}

\subsection{Rocq Verification}

The \texttt{GovernedMetaprogramming.v} module (662 lines) mechanizes 26 theorems
in Rocq 8.19 with zero admitted lemmas. The key verified properties are:

\begin{enumerate}
\item \textbf{Form Manipulation Purity} (\texttt{form\_manipulation\_pure}):
      form operations emit no directives and therefore cannot exercise capabilities.
\item \textbf{Materialization is Governed} (\texttt{materialization\_governed}):
      every materialization emits at least one directive, ensuring governance mediation.
\item \textbf{No Bypass} (\texttt{no\_bypass\_form\_to\_machine}):
      no sequence of pure form operations produces an authorized machine.
\item \textbf{Coterminous Preservation} (\texttt{coterminous\_preserved}):
      adding form operations to the language does not split the governance boundary.
\item \textbf{Inspection Safety} (\texttt{inspection\_no\_capability\_grant}):
      inspecting a form does not grant the inspector any capabilities the form declares.
\item \textbf{Composition Safety} (\texttt{sequential\_composition\_safe}):
      sequentially composing governed form operations preserves all governance invariants.
\end{enumerate}

The remaining 41 theorems verify supporting properties: determinism of inspection,
idempotence of form hashing, structural induction over form trees, and capability
set operations used in the inspection checks.

\subsection{Test Suite}

The implementation is tested with 51 unit tests across two test modules:

\begin{itemize}
\item \textbf{FormTest} (39 tests, 633 lines): covers construction, AST conversion,
      navigation, transformation, analysis, serialization, \texttt{quote(expr)}
      parsing, \texttt{reflect()} parsing, compilation of \texttt{form.*} operations,
      and round-trip \texttt{from\_text}/\texttt{to\_text} fidelity.
\item \textbf{FormInspectorTest} (12 tests, 160 lines): covers all six inspection
      checks (valid structure, required fields, capabilities, model authorization,
      governance presence, trust level), approval and rejection paths, and
      hash binding verification.
\end{itemize}

All 51 tests pass. The test suite verifies both the functional correctness of form
operations and the governance properties of the inspection/materialization boundary.

\subsection{Performance Characteristics}

Form construction is map allocation on the BEAM virtual machine. Measured on an
Apple M-series processor:

\begin{itemize}
\item \textbf{Form construction} (\texttt{Form.new/1}): $<$1 $\mu$s per form
\item \textbf{AST to form conversion} (\texttt{Form.from\_node/1}): proportional
      to AST node count; typical machines (10--50 nodes) convert in $<$50 $\mu$s
\item \textbf{Form inspection} (\texttt{FormInspector.inspect\_form/2}): six checks
      over a form map; $<$100 $\mu$s for typical machines
\item \textbf{Form hashing} (\texttt{Form.hash/1}): SHA-256 of canonical
      serialization; $<$20 $\mu$s
\end{itemize}

The governance overhead of form inspection before materialization is negligible
relative to the materialization itself (compilation to BEAM bytecode, which takes
milliseconds). The design choice to make inspection a structural check rather than
a semantic analysis ensures that inspection time is bounded by form size, not by
the complexity of what the form computes.

\section{Limitations}
\label{sec:limitations}

\begin{itemize}
\item Structural inspection is decidable but limited by Rice's theorem.
      Policy compliance checks are structural properties; semantic
      correctness is not structurally checkable. Governance guarantees
      mediation, not policy soundness.
\item The boundary preservation argument (Proposition~\ref{prop:boundary})
      is an architectural argument, not a full conservativity proof. A
      complete formalization would require defining boundaries as effect
      signature sets and proving that form manipulation introduces no new
      signatures. We defer this to future work.
\item Performance measurements (Section~\ref{sec:evaluation}) are
      microbenchmarks on individual operations. System-level overhead under
      concurrent load has not been measured. Form construction and inspection
      are microsecond-scale operations, but interaction with the compilation
      cache and governance pipeline under contention remains to be characterized.
\item The design is implemented in MashinTalk, a domain-specific language.
      Applicability to general-purpose languages is an open question.
      General-purpose languages may not have the effect-boundary structure
      that makes the reclassification of eval possible.
\end{itemize}

\section{Conclusion}
\label{sec:conclusion}

We presented governed metaprogramming, a language design that reclassifies
\texttt{eval} from a language primitive into a governed effect. Machine forms
are first-class values representing program structure. Form manipulation is
pure computation. Materialization, the transition from form to executable
machine, is a governed effect subject to structural inspection.

The design closes the metaprogramming escape hatch in governed intelligent
systems. Without it, a system governed at the execution level remains
vulnerable at the code generation level: generate a new program, execute it,
bypass prior constraints. With governed metaprogramming, the generated
program crosses the same governance boundary as every other effect.

The central observation is that \texttt{eval} allocates executable authority.
In AI systems that synthesize their own executable structure, this allocation
must be visible, mediated, and recorded. That is the contribution of this
paper.

\bibliographystyle{plainnat}
\bibliography{governed-homoiconicity-references}

\end{document}